\documentclass[a4paper,conference]{IEEEtran}
\IEEEoverridecommandlockouts

\usepackage{amsmath,amsthm,amsfonts,amssymb}
\usepackage{algorithmic,algorithm}
\usepackage{array,soul,xcolor}
\usepackage{mathtools,cases,relsize}
\usepackage{subcaption}
\usepackage{textcomp,xpatch}
\usepackage{stfloats,subfiles}
\usepackage{multirow}
\usepackage{graphicx,subcaption}
\usepackage{cite,url,verbatim}

\newtheoremstyle{colon}%
{}
{}
{\itshape}
{}
{\bfseries}
{:}
{ }
{}
\theoremstyle{colon}
\makeatletter
\xpatchcmd{\proof}{\@addpunct{.}}{\@addpunct{:}}{}{}
\makeatother
\makeatletter
\xpatchcmd{\proof}{\@addpunct{.}}{\@addpunct{:}}{}{}
\makeatother

\newtheorem{theorem}{Theorem}

\newtheorem{remark}{Remark}

\newcommand{\bs}{\text{BS}}
\newcommand{\cu}{\text{CU}}
\newcommand{\uk}{U_k}
\newcommand{\bd}{\text{BD}}
\newcommand{\FoxH}{$\mathrm{H}^{\cdot}_{\cdot}
      \left(\cdot  \left|\begin{smallmatrix} \dotsb\\ \dotsb
      \end{smallmatrix}\right.  \right)\;$}
\newcommand{\BiHfox}{$\mathrm{H}^{\cdot;\cdot;\cdot}_{\cdot;\cdot;\cdot}
      \left( \begin{smallmatrix}\dotsb\\\dotsb
      \end{smallmatrix}\left|\begin{smallmatrix} \dotsb\\ \dotsb
      \end{smallmatrix}\right.
      \left| \begin{smallmatrix} \dotsb\\\dotsb
      \end{smallmatrix}\right.  \left|
      \cdot,\cdot\right.
      \right)\;$}
\begin{document}
\title{\huge{A Novel Paradigm Shift for Next-Generation: Symbiotic Backscatter Rate-Splitting Multiple Access Systems}

    \thanks{This work was supported by the Research Program through the National Research Foundation of Korea under Grant NRF-2023R1A2C1003546.}
}
\author{\IEEEauthorblockN{Thai-Hoc Vu$^{1}$, Daniel Benevides da Costa$^{2}$, Bao Vo Nguyen Quoc$^{3}$, and Sunghwan Kim$^{1}$}
\IEEEauthorblockA{
$^{1}$Department of Electrical, Electronic and Computer Engineering, University of Ulsan, Republic of Korea\\
$^{2}$Department of Electrical Engineering, King Fahd University of Petroleum $\&$ Minerals, Dhahran 31261, Saudi Arabia \\
$^{3}$Van Lang School of Technology, Van Lang University, Vietnam\\
Emails: vuthaihoc1995@gmail.com, danielbcosta@ieee.org, bao.vnq@vlu.edu.vn, sungkim@ulsan.ac.kr}}
\maketitle

\begin{abstract}
Next-generation wireless networks are projected to empower a broad range of Internet-of-things (IoT) applications and services with extreme data rates, posing new challenges in delivering large-scale connectivity at a low cost to current communication paradigms. Rate-splitting multiple access (RSMA) is one of the most spotlight nominees, conceived to address spectrum scarcity while reaching massive connectivity. Meanwhile, symbiotic communication is said to be an inexpensive way to realize future IoT on a large scale.  To reach the goal of spectrum efficiency improvement and low energy consumption, we merge these advances by means of introducing a novel paradigm shift, called symbiotic backscatter RSMA, for the next generation. Specifically, we first establish the way to operate the symbiotic system to assist the readers in apprehending the proposed paradigm, then guide detailed design in beamforming weights with four potential gain-control (GC) strategies for enhancing symbiotic communication, and finally provide an information-theoretic framework using a new metric, called symbiotic outage probability (SOP) to characterize the proposed system performance. Through numerical result experiments, we show that the developed framework can accurately predict the actual SOP and the efficacy of the proposed GC strategies in improving the SOP performance.
\end{abstract}

\section{Introduction}
Recent advances powered by the Internet-of-Things (IoT) drive human lives more convenient with emerged applications (e.g., home automation or driverless vehicles) and services (e.g., telemetry, healthcare, or the metaverse). Particularly, communication over the air plays a key role in managing and coordinating information between entities. However, the headway of IoT brings major challenges to the development of next-generation wireless infrastructure, of which spectrum scarcity and energy consumption are two critical issues  \cite{DeAlwis2021Apr}.
 
In the context of spectrum scarcity, several technologies have been introduced to overcome this challenge, one of which is non-orthogonal multiple access (NOMA). Compared to classical orthogonal multiple access (OMA), NOMA benefits the system in terms of time and frequency resources. However, the complex process of pairing users, the need to share user signals, and the unfairness in successive interference cancellation (SIC) approaches make it impractical. Despite significant efforts to promote NOMA for the next generation through 3GPP proposals \cite{Makki2020Jan}, the emergence of rate-splitting multiple access (RSMA) nominees in recent years has led to this technology gradually being out of the game. Specifically, RSMA brings more pleasing dominant features than NOMA, such as increased achievable rate, fairness in SIC use, and user confidentiality \cite{Clerckx2023Feb}. This is achieved by exploiting message separation and multiplexing signals in the power domain at the transmitter while leveraging the attribute brought by SIC to enable the receivers to decode the common message and their individual ones without the information knowledge of the others as in NOMA. RSMA has been therefore devised as the limelight candidate for the 6G, with a series of studies for joint radar and communications \cite{Xu2021Sep}, aerial-assisted terrestrial networks \cite{VuAug}, intelligent reflecting surface \cite{Hua2023May}, and Terahertz communication \cite{Vu2023Sep}. Particularly, the feasibility of RSMA has been verified with the first-ever prototype in \cite{Lyu2024Mar}.

On the other hand, research on green wireless communication, especially in IoT networks, has been growing due to the increasing popularity of limited-energy IoT devices. In light of radio frequency (RF) availability, affordable energy harvesting solutions, represented by wireless power transfer communication networks and simultaneously wireless information power transfer \cite{Lu2014Nov}, are familiar but often consume high power, making them difficult to integrate into large-scale IoT. To sidestep energy limitations coupled with severe spectrum shortages, symbiotic radio (SR) have emerged as a new breath of fresh air for the passive IoT era \cite{Long2019Nov}. In particular, this  paradigm allows backscatter devices ($\bd$), which are live parasitic in the primary network, to exchange information with their secondary receiver via the legitimate frequency band without energy-cost payment. Hence, the research on SR adoptions has received significant interest in recent years, with studies on SR communication and sensing \cite{Ren2023Mar}, static SR NOMA \cite{Zhang2019Feb}, and dynamic SR NOMA \cite{Yang2023Mar}. 

In this paper, we introduce a novel paradigm shift for next-generation wireless communication systems to enhance spectrum utilization, referred to as symbiotic backscatter RSMA systems of cellular and IoT networks. In this paradigm, base-station ($\bs$) uses RSMA signalling to serve cellular users ($\cu$s) in the same resource block simultaneously, whereas a $\bd$ passively modulates its own information by switching its antenna impedance and reflecting the free-RF signal emitted by $\bs$. From a point of view, the proposed paradigm lines two typical systems: RSMA for cellular and sensing for IoT networks. Interestingly, if the $\bd$ does not work, the proposed system turns into the conventional RSMA system. However, the proposed system differs from traditional RSMA and ambient backscatter communication (BackCom) systems. Therefore, the existing measurement metrics, optimization algorithms, and analytical frameworks cannot be directly applied to evaluate the proposed symbiotic backscatter RSMA system. To the best of the authors' knowledge, investigations of this novel paradigm have not been discussed in the literature in terms of both analysis and optimization. These motivate us to delve into finding the answer to three questions: 
\begin{enumerate}
    \item \textit{How to establish symbiotic RSMA communication?}
    \item \textit{How to provide sustainable symbiotic communication?}
    \item \textit{How to evaluate the symbiotic system's performance?}
\end{enumerate}

With the efforts in enhancing next-generation wireless communication systems, our main contribution in this paper can be summarized as follows:
\begin{itemize}
    \item A novel paradigm shift for the next generation, named symbiotic backscatter RSMA, is proposed to not only enhance spectrum utilization but also provide a sustainable communication medium with low energy consumption.

    \item A novel beamforming approach is designed to facilitate the symbiotic of two distinct networks, followed by four gain control (GC) strategies based on channel gain assignment to boost the performance of the proposed system, including random channel selection (RCS), smallest channel selection (SCS), maximal channel selection (MCS), and composite channel selection (CCS).

    \item An information-theoretic framework using a new metric, called symbiotic outage probability (SOP),  is provided for evaluating the proposed system's performance.

    \item Numerical results are presented to validate the developed framework, to compare the performance of four proposed strategies, as well as to inspect the impact of the systems' key parameters on the SOP performance.
    
\end{itemize}


\textit{Notations}: \FoxH 
 denotes the Fox-H function \cite[eq. (9.301)]{integral}. \BiHfox 
 denotes the generalized bivariate Fox-H function \cite[eq. 2.2.1]{ismail1980h}.  $\Gamma(\cdot,\cdot)$ denotes the upper incomplete Gamma function \cite[eq. (8.350.2)]{integral}.    

\section{System Model}\label{sec2}
\begin{figure}
    \centering
    \includegraphics[width =0.9\linewidth]{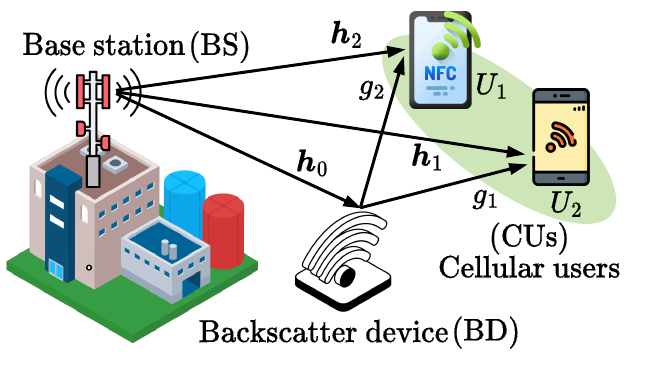}
    \caption{Illustration of the system model.}
    \label{fig1}
\end{figure}

Fig.~\ref{fig1} illustrates a symbiotic network composed of one $\bs$ equipped with $L\ge 2$ antennas, two single-antenna $\cu$s, and one $\bd$ having signal processing ability. The $\bs$ serve two $\cu$s through RSMA signaling, whereas $\bd$ passively communicates by reflecting and modulating the received incident RF signal emitted by $\bs$ to ride on its information. The considered system acts in quasi-static Rayleigh block fading with the same frequency band for both uplink and downlink. 

Denote by $\mathcal{L}= \{1,...,L\}$ and $\mathcal{K}= \{1,2\}$ the set of $\bs$ antennas and $\cu$s, respectively, while $\uk$ is the $k$-th $\cu$, with $k\in\mathcal{K}$. For each fading block, the channel coefficient vectors of links $\bs \to \bd$, $\bs \to \uk$, and $\bd \to \uk$ are given by $\boldsymbol{h}_0 \in \mathbb{C}^{L\times 1}$, $\boldsymbol{h}_k \in \mathbb{C}^{L\times 1}$, and $g_k$, respectively. Under Rayleigh mediums, $|h_{0,l}|^2$, with $l\in\mathcal{L}$, $|h_{k,l}|^2$, and $|g_k|^2$ under go the exponential distribution with the respective parameters $\lambda_0 $, $\lambda_k$, and $\Omega_k$. The probability density function (PDF) of $|g_k|^2$, for example, is given by $f_{|g_k|^2}(z) = \Omega_k\exp\left(-\Omega_kz\right), \forall z\ge 0.$

Under quasi-static fading block conditions, the channel state information of $\boldsymbol{h}_0$, $\boldsymbol{h}_k$, and $g_k$ can be easily estimated thanks to two phases of uplink pilot training. In the first phase, $\bd$ absorbs the signal completely with reflection, facilitating $\bs$ to attain $\boldsymbol{h}_k$. In the second phase, $\bd$ enters the fixed impedance mode with predefined symbol $c_b$, the backscatter efficiency $\eta\in(0,1]$, and reflection coefficient $\delta\in(0,1]$. Using the training pilots, $\bs$ can obtain $\boldsymbol{h}_0g_k$ by subtracting $\boldsymbol{h}_k$ from the composite channel $\boldsymbol{h}_k + \sqrt{\zeta\delta}c_b\boldsymbol{h}_0g_k$. Afterwards, $\bs$ can acquire $\boldsymbol{h}_0$ using the expectation-maximization algorithm. Finally, $\boldsymbol{h}_0$ and $\boldsymbol{h}_k$ will be used to design beamforming weights, resource allocation, and data decoding at $\cu$s.

\subsection{Transmission Model}
Initially, $\bs$ divides $\cu$s' messages into the common and private parts, then encodes all common parts into a unique stream $s_c$, and finally encodes the private ones into the individual streams $s_k$ separately, with $\mathrm{E}\{|s_c|^2\} =\mathrm{E}\{|s_k|^2\} =1$. Let $p$ denotes the transmit power of the $\bs$ with power allocation (PA) factors $\alpha_c$ and $\alpha_k$ to $s_c$ and $s_k$, respectively, with  $\alpha_c + \sum_{k\in\mathcal{K}}\alpha_k=1$. Then, the superposition message sent by $\bs$ at the $t$-th time slot can be written as 
\begin{align}
\label{eq:01}
    \boldsymbol{s}(t) = \boldsymbol{w}_cs_c(t) \sqrt{\alpha_c p} + \sum\nolimits_{k}\boldsymbol{w}_ks_k(t)\sqrt{\alpha_k p},
\end{align}
where $\boldsymbol{w}_c \in \mathbb{C}^{L\times 1} $ and $\boldsymbol{w}_k \in \mathbb{C}^{L\times 1}$ are the beamforming weights. At $\uk$, the composite signal from the direct and backscatter links under the additive white Gaussian noise $u_k(t)$ with zero-mean and variance $\sigma^2$ can be written as
\begin{align}
\label{eq:02}
    y_k(t) = \boldsymbol{h}_k^H\boldsymbol{s}(t) + \sqrt{\eta\delta}\boldsymbol{h}_0^H\boldsymbol{s}(t)g_kc_b(t) + u_k(t).
\end{align}

Normally, the backscattered signal from $\bd$ is subject to multiplicative fading from the cascaded link $\boldsymbol{h}_0^H\boldsymbol{s}(t)g_k$ and power losses in reflections. Thus, it is feasible for $\uk$ to decode $s_c(t), s_k(t)$, and $c_b(t)$ consecutively with the aid of SIC mechanisms. For notational brevity, we denote $j\in\{\mathcal{K}/k\}$, $v_c \triangleq \alpha_c|\boldsymbol{h}_k^H\boldsymbol{w}_c|^2$, $v_k \triangleq \alpha_k|\boldsymbol{h}_k^H\boldsymbol{w}_k|^2$, $v_j \triangleq \sum\nolimits_{j} \alpha_j|\boldsymbol{h}_k^H\boldsymbol{w}_j|^2$, $v_{0,k} \triangleq   \alpha_c|\boldsymbol{h}_0^H\boldsymbol{w}_c|^2 + \sum_{k\in\mathcal{K}}\alpha_k|\boldsymbol{h}_0^H\boldsymbol{w}_k|^2$. The \textit{signal-to-interference-plus-noise ratios} for decoding $s_c(t)$, $s_k(t)$, and $c_b(t)$ can be, respectively, written as
\begin{align}
\label{eq:03}
    \gamma_{k,c} & =\frac{pv_c}{p(v_k + v_j + \eta\delta v_{0,k}|g_k|^2) + \sigma^2},\\
\label{eq:04}
    \gamma_{k,k}  &= \frac{pv_k}{p( v_j  + \eta\delta v_{0,k}|g_k|^2) + \sigma^2},  
    \gamma_{k,b}  = \frac{ p \eta\delta v_{0,k}|g_k|^2 }{p v_j  + \sigma^2}.
\end{align}
It is intuitively found that: \textbf{i)} the presence of the component $v_j$ is harmful to decode $s_c(t)$, $s_k(t)$, and $c_b(t)$; \textbf{ii)} Maximizing the components $v_c$ and $v_k$ respectively enhances $\gamma_{k,c}$ and $\gamma_{k,k}$, thereby increasing the decoding ability of $s_c(t)$ and $s_k(t)$; and \textbf{iii)} Boosting $v_{0,k}$ helps improves $\gamma_{k,b}$, increasing the decoding ability of $c_b(t)$. However, at the same time, it also decreases $\gamma_{k,c}$ and $\gamma_{k,k}$, deteriorating the decoding ability of $s_c(t)$ and $s_k(t)$. To overcome these circumvents, we are excited to investigate designing the beamforming weights.

\subsection{Beamforming Weight Design and Gain Control}

\textit{\textbf{General Methodology}}: From the above observations, we start with the there unconstrained optimization problems:  
\begin{align}
\label{eq:05}
   &\mathrm{Pro}_1: \max_{\boldsymbol{w}_c} {v_c},\quad \mathrm{Pro}_2: \max_{\boldsymbol{w}_k} {v_k},\quad \mathrm{Pro}_3: \min_{\boldsymbol{w}_j} {v_j}.
\end{align}
To deal with the problems in \eqref{eq:05}, let us inspect $v_c$, $v_k$, and $v_j$.  The results show that $v_c$ and $v_k$ are maximized and $v_j$ is minimized if and only if the following results are satisfied
\begin{align}
    |\boldsymbol{h}_k^H\boldsymbol{w}_c| = \|\boldsymbol{h}_k\|, \quad |\boldsymbol{h}_k^H\boldsymbol{w}_k| = \|\boldsymbol{h}_k\|, \quad |\boldsymbol{h}_k^H\boldsymbol{w}_j| = 0.
\end{align}
Building upon zero-forcing (ZF) and maximum ratio transmission (MRT) approaches,  $\boldsymbol{w}_k$ and $\boldsymbol{w}_c$ can be developed as
\begin{align}
\label{eq:07}
    \boldsymbol{w}_k = \widetilde{\boldsymbol{w}}_k\|\boldsymbol{h}_k\|, \quad \boldsymbol{w}_c = \boldsymbol{w}_1 +\boldsymbol{w}_2,
\end{align}
where $\widetilde{\boldsymbol{w}}_k$ is the $k$-th column vector of a weight matrix $\boldsymbol{W}\in\mathbb{C}^{L\times 2}$, with capable of $\boldsymbol{h}_k^H\widetilde{\boldsymbol{w}}_k = 1$ and $\boldsymbol{h}_k^H\widetilde{\boldsymbol{w}}_j = 0$, followed by the add-on MRT component $\|\boldsymbol{h}_k\|$ to maximize $v_c$ and $v_k$.

Next, inspections of $v_{0,k}$ disclose that $\gamma_{k,c}$ and $\gamma_{k,k}$ only improve if $v_{0,k}$ goes to zero. However, this declines BackCom's benefits, asking for mitigating one of two components $r_k \triangleq\sum_{k}\alpha_k|\boldsymbol{h}_0^H\boldsymbol{w}_k|^2$ and $r_c\triangleq \alpha_c|\boldsymbol{h}_0^H\boldsymbol{w}_c|^2$ from $v_{0,k}$, while finding an alternative to adjust the rest. As a matter of fact, eliminating $r_k$ improves  $\gamma_{k,c}$ and $\gamma_{k,k}$ better than $r_c$, which yields the necessity of $|\boldsymbol{h}_0^H\boldsymbol{w}_k| = 0$. Hence, the structure of $\boldsymbol{W}$ is no longer to $L\times 2$ but instead extents to $L\times 3$, i.e.,
\begin{align}
    \boldsymbol{W} = [\widetilde{\boldsymbol{w}}_0, \widetilde{\boldsymbol{w}}_1,\widetilde{\boldsymbol{w}}_2] = \boldsymbol{H}^H(\boldsymbol{H}\boldsymbol{H}^H)^{-1},
\end{align}
where $\boldsymbol{H} = [\boldsymbol{h}_0,\boldsymbol{h}_1,\boldsymbol{h}_2]$. Unfortunately, the use of $\boldsymbol{w}_c$ in \eqref{eq:07} is now no longer valid since $r_c = 0$. Knowing that $|\boldsymbol{h}_0^H\boldsymbol{w}_0| = 1$, one can arrive at the new solution for $\boldsymbol{w}_c$ as
\begin{align}
    \label{eq:09}\boldsymbol{w}_c = \boldsymbol{w}_0 + \boldsymbol{w}_1 + \boldsymbol{w}_2,
\end{align}
where $\boldsymbol{w}_0 = \widetilde{\boldsymbol{w}}_0 \theta$, with $\theta$ being an add-on GC factor. 
\begin{remark}
    Note that the developed solution above is still applicable for multi-users case (i.e., $K\ge 2$) by scaling up $\boldsymbol{H} =[\boldsymbol{h}_0,\boldsymbol{h}_1,\dotsb,\boldsymbol{h}_K]$, yielding $\boldsymbol{W} = [\widetilde{\boldsymbol{w}}_0, \widetilde{\boldsymbol{w}}_1,\dotsb,\widetilde{\boldsymbol{w}}_K]$. However, this would result in a higher computational complexity for the systems. Therefore, scaling down multiple users into multiple clusters with a smaller number of users served, along with the use of hybrid OMA-RSMA, is necessary. 
\end{remark}

\textit{\textbf{GC Selection}}: With the approach above, it is interesting to further study how to design $\theta$. In what follows, we introduce four GC strategies based on channel gain assignments:
\begin{align}
\label{eq:10}
    \text{RCS}: \theta_1 &= \underset{l}{\text{ran}} |h_{0,l}| ,& \quad 
    \text{SCS}: \theta_2 &= \min_{l} |h_{0,l}|,\\
\label{eq:11}
    \text{MCS}: \theta_3 &= \max_{l} |h_{0,l}|,&\quad
    \text{CCS}: \theta_4 &= \sum_{l} |h_{0,l}|.
\end{align}

\begin{remark}\label{remark2}
    It is important to note that exploiting the strategies above comes with certain trade-offs in the system performance. RCS is the simplest method to generate $\theta$ with negligible impact on $\gamma_{k,c}$, $\gamma_{k,k}$, and $\gamma_{k,b}$. Nevertheless, it is strongly pronounced with SCS, MCS, and CCS. Particularly, SCS is the worst choice to $\gamma_{k,b}$, even though it shows the lowest interference to $\gamma_{k,c}$ and $\gamma_{k,k}$.  Meanwhile, MCS is an excellent option for enhancing $\gamma_{k,b}$ while declining $\gamma_{k,c}$ and $\gamma_{k,k}$. Similarly, CCS provides the most significant improvement to $\gamma_{k,b}$ but also reduces $\gamma_{k,c}$ and $\gamma_{k,k}$ significantly.  
\end{remark}

\section{Information-Theoretic Framework: Symbiotic Outage Probability (SOP) Metric}\label{sec3}

Prior to going into the detailed analysis, let us denote $\Psi = p/\sigma^2$ the average signal-to-noise ratio (SNR) [dB]. Next, injecting the developed weight vectors $\boldsymbol{w}_c$ in \eqref{eq:09} and $\boldsymbol{w}_k$ in \eqref{eq:07} into the formulas \eqref{eq:03} and \eqref{eq:04}  return the following results
\begin{align}
\label{eq:12}
    \gamma_{k,c} &= \frac{\alpha_c\Psi\tau_k}{\alpha_k\Psi\tau_k  + \eta\delta\Psi |\theta|^2|g_k|^2 + 1}, (\tau_k \triangleq \|\boldsymbol{h}_k\|^2)\\
\label{eq:13}
    \gamma_{k,k} &= \frac{\alpha_k\Psi\tau_k}{\eta\delta\Psi |\theta|^2|g_k|^2  + 1}, \quad
    \gamma_{k,b} =  \eta\delta\Psi |\theta|^2|g_k|^2.
\end{align}

Having obtained $\gamma_{k,c}$, $\gamma_{k,k}$, and $\gamma_{k,b}$, it is exciting to design $\delta$ to ensure simultaneous transmission of RSMA and Backcom signals, i.e., $s_c(t)$, $s_k(t)$, and $c_b(t)$ can be successfully decoded at $\uk$. Subsequently, the SIC procedure should satisfy the rule: $\gamma_{k,c} > \bar{\gamma}_{c}$, $\gamma_{k,k} > \bar{\gamma}_{k}$, and $\gamma_{k,b} > \bar{\gamma}_{b}$, where $\bar{\gamma}_{c}=2^{R_c}-1$, $\bar{\gamma}_{k} =2^{R_k}-1$, and $\bar{\gamma}_{b} = 2^{R_b}-1$ represent the respective thresholds to decode $s_c(t)$, $s_k(t)$, and $c_b(t)$, with $R_c,R_k,$ and $R_b$ being the respect transmission rate targets [bps/Hz]. Consider all $\cu$s, the range of $\delta$ can be written as
\begin{align}
\label{eq:14}
   \max\left\{\boldsymbol{\delta}_{\mathrm{low}}\right\}  < \delta < \min \left\{1,\boldsymbol{\delta}_{\mathrm{up}}\right\}, \forall k\in\mathcal{K},
\end{align}
where $\boldsymbol{\delta}_{\mathrm{up}} = \left\{ \frac{\rho_{c,k}\Psi\tau_k - 1}{\eta\Psi|\theta|^2|g_k|^2}, \frac{\rho_k\Psi\tau_k - 1}{\eta\Psi |\theta|^2|g_k|^2  }  \right\}$, with $\rho_{c,k} \triangleq (\alpha_c - \bar{\gamma}_{c}\alpha_k)/\bar{\gamma}_{c}$ and $\rho_k \triangleq \alpha_k/\bar{\gamma}_{k}$, while $\boldsymbol{\delta}_{\mathrm{low}} = \left\{\frac{\bar{\gamma}_{b}}{\eta\Psi |\theta|^2|g_k|^2} \right\}$.

Due to page constraints, analyzing all cases of $|\theta|^2$ as developed in \eqref{eq:10}-\eqref{eq:11} is infeasible. Rather, we study the most challenging case  $|\theta|^2 =  \sum_{l\in\mathcal{L}} |h_{0,l}| \triangleq \tau_0$ while treating the other cases using simulation methods (readers can apply the same method below to the remaining cases).  
Therewith, we turn to express the cumulative distribution function (CDF) of $\tau_i$, $i=\{0,k\}$ based on moment matching approach \cite{Vu2023Sep}, i.e., 
 \begin{align}
 \label{eq:15}
     F_{\tau_i}(z) = 1 \!- \frac{\Gamma(L,\lambda_i z)}{(L-1)!} \!= 1 -\sum_{m=0}^{L-1}(\lambda_i z)^m\frac{\exp(-\lambda_i z)}{m!}.
\end{align}
 
Having \eqref{eq:14} and \eqref{eq:15} in hand, the next focus is on evaluating the SOP of the system. In terms of definition, the SOP occurs when any condition of the SIC rule becomes invalid. Accordingly, the SOP can be mathematically expressed as 
\begin{align}
\label{eq:16}
    \mathrm{SOP}
    &= 1 - \Pr\Big[\max\left\{\boldsymbol{\delta}_{\mathrm{low}}\right\}  <  \min \left\{1,\boldsymbol{\delta}_{\mathrm{up}}\right\}\Big]\nonumber\\
    &= 1 - \Pr\left[\frac{\bar{\gamma}_{b}}{\eta\Psi \tau_0|g_j|^2} \! <  \min \left\{1,\boldsymbol{\delta}_{\mathrm{up}}\right\},|g_k|^2>|g_j|^2\right]\nonumber\\
    &- \Pr\bigg[\frac{\bar{\gamma}_{b}}{\eta\Psi \tau_0|g_k|^2} \! <  \min \left\{1,\boldsymbol{\delta}_{\mathrm{up}}\right\},|g_j|^2 \!>|g_k|^2\bigg].
\end{align}

\begin{theorem}\label{Theorem1}
Exact closed-form solution for the SOP of the system can be formulated in \eqref{eq:17} on the top page, where $\beta_0 \triangleq {\bar{\gamma}_{b}\lambda_0}{/[\eta\Psi]}$, $\beta_k \triangleq {\lambda_k\pi_k\bar{\gamma}_{b}}{/[\Psi\Omega_k ]}$, $\beta_{j,k} \triangleq \Omega_ j + \Omega_k$, $\chi_k \triangleq \beta_k\Omega_k(\bar{\gamma}_{b}+1)/{\bar{\gamma}_{b}}$, and $\pi_k \triangleq \max\{1/\rho_k,1/\rho_{c,k}\}$. 
    \begin{figure*}
        \begin{align}
        \label{eq:17}
        \mathrm{SOP} =   1 &- \sum_{m=0}^{L-1}\sum_{n=0}^{L-1}\sum_{q=0}^{n}\frac{\Omega_1\beta_0^m \chi_2^q\exp\left( -\chi_2\right)\beta_{1,2}^{m-2}}{m! q!\beta_2\Gamma(n-q+1)} 
        \mathrm{H}^{0,1;1,1;0,1}_{1,0;1,1;1,0}\left(\begin{matrix}
          m-1;1,1\\
          -
      \end{matrix} 
      \left|\begin{matrix}
          (q- n,1)\\
          (0,1)
      \end{matrix} \right.
      \left|\begin{matrix}
          (1,1)\\
          -
      \end{matrix} \right. \left| \frac{1}{\beta_{2}\beta_{1,2}};\frac{1}{\beta_{0}\beta_{1,2}} \right. \right)
       \frac{1}{(L-1)!}\nonumber\\
      &\times \Gamma\left(L,( \bar{\gamma}_{b} + 1)\frac{\lambda_1 \pi_1}{\Psi}\right) 
      - \sum_{m=0}^{L-1}\sum_{n=0}^{L-1}\sum_{q=0}^{n}\frac{\Omega_2\beta_0^m \chi_1^q\exp\left( -\chi_1\right)\beta_{2,1}^{m-2}}{m! q!\beta_1\Gamma(n-q+1)} \frac{1}{(L-1)!} \Gamma\left(L,( \bar{\gamma}_{b} + 1)\frac{\lambda_2 \pi_2}{\Psi}\right)
    \nonumber\\
      & \times         \mathrm{H}^{0,1;1,1;0,1}_{1,0;1,1;1,0}\left(\begin{matrix}
          m-1;1,1\\
          -
      \end{matrix} 
      \left|\begin{matrix}
          (q- n,1)\\
          (0,1)
      \end{matrix} \right.
      \left|\begin{matrix}
          (1,1)\\
          -
      \end{matrix} \right. \left| \frac{1}{\beta_{1}\beta_{2,1}};\frac{1}{\beta_{0}\beta_{2,1}} \right. \right)
       .
    \end{align}
   \hrulefill
    \end{figure*}
    
\end{theorem}

\begin{proof}
From  \eqref{eq:16}, the first probability (the second probability can be done similarly by swapping the position of $j$ and $k$ with each other) can be arranged in an elegant form as  
\begin{align}
\label{eq:18}
    \Phi &=\Pr\left[
         \tau_0|g_j|^2 \ge \frac{\bar{\gamma}_{b}}{ \eta\Psi},
       \tau_k \ge  \left(\frac{|g_k|^2}{|g_j|^2}\bar{\gamma}_{b} + 1\right)\frac{\pi_k}{\Psi},
       \frac{|g_k|^2}{|g_j|^2}>1
    \right]\nonumber\\
    &\quad\times\Pr\left[\tau_j  \ge ( \bar{\gamma}_{b} + 1)\frac{\pi_j}{\Psi}\right].
\end{align}
To proceed, making use of \eqref{eq:15} and the PDF of $|g_k|^2$, the first probability in  \eqref{eq:18} can be computed as
\begin{align}
\label{eq:19}
    I_1 &= \int\limits_{0}^{\infty}\left[1-F_{\tau_0}\left(\frac{\bar{\gamma}_{b}}{x\eta\Psi}\right)\right]\int\limits_{x}^{\infty}\left[1-F_{\tau_k}\left(\left[\frac{y}{x}\bar{\gamma}_{b}+1\right]\frac{\pi_k}{\Psi}\right)\right]\nonumber\\
    &\quad\times f_{|g_j|^2}(x)f_{|g_k|^2}(y)dxdy\nonumber\\
    &{=}\sum_{m=0}^{L-1}\sum_{n=0}^{L-1}\frac{\Omega_j\beta_0^m\Omega_k}{m! n!\bar{\gamma}_{b}} \left(\frac{\lambda_k\pi_k}{\Psi}\right)^n\int\limits_{0}^{\infty} \exp\left(-\frac{\beta_0}{x} + \frac{\Omega_k}{\bar{\gamma}_{b}}x \right)\nonumber\\
    &\quad
     \times\frac{\exp\left(-\Omega_jx \right) }{x^{m-1}}\left[\int_{\bar{\gamma}_{b}+1}^{\infty}\hspace{-7.5 pt}z^n \exp\left(  -\frac{\Omega_k}{\bar{\gamma}_{b}}\left[x+\beta_k\right] z\right)\! dz\right]\! dx\nonumber\\
    &{=}\sum_{m=0}^{L-1}\sum_{n=0}^{L-1}\sum_{q=0}^{n}\frac{\Omega_j\beta_0^m}{m! q!}\beta_k^n \left(\frac{\bar{\gamma}_{b}+1}{\bar{\gamma}_{b}/\Omega_k}\right)^q\!\exp\left( -\beta_k\frac{\bar{\gamma}_{b}+1}{\bar{\gamma}_{b}/\Omega_k}\right) \nonumber\\
    &\quad\times\underbrace{\int_{0}^{\infty} \frac{x^{-m+1}}{\left(x+\beta_k\right)^{n-q+1}} \exp\left(-\frac{\beta_0}{x}-\beta_{j,k}x \right)
      dx}_{\triangleq\; \Xi},
  \end{align}  
where second step is derived by using variable transform $z = y\bar{\gamma}_{b}/x +1$ and last step is obtained thanks to \cite[eq. (351.2)]{integral}. 

Direct assessing $\Xi$ is not straightforward. Thereon, we propose to address this problem by conjuring the three following transformations into the Fox-H functions \cite[Appendix A7 and eq. (1.2.2)]{ismail1980h}: $\exp(-z)  = \mathrm{H}^{1,0}_{0,1}\big[z \big|\begin{smallmatrix}-\\  (0,1) \end{smallmatrix} \big]$, $\exp\left(-\frac{1}{z}\right)   = \mathrm{H}^{0,1}_{1,0}\big[z \big|\begin{smallmatrix} (1,1)\\ - \end{smallmatrix} \big]$, and $\left(1 + cz\right)^{-m} = \frac{1}{\Gamma(m)}\mathrm{H}^{1,1}_{1,1}\big[\frac{z}{c} \big|\begin{smallmatrix} (1-m,1)\\(0,1) \end{smallmatrix}\big]$. 
It is due to the highly general nature of this function representing plenty of elementary functions, lending itself to concise integral formulas.
Therewith, $\Xi$ can be rewritten as
\begin{align}
\label{eq:20}
  \Xi   
    &=\frac{\beta_k^{-(n-q+1)}}{\Gamma(n-q+1)}\int_{0}^{\infty} x^{-m+1}    \mathrm{H}^{1,0}_{0,1}\left(\beta_{j,k} x \left|\begin{matrix}
          -\\
          (0,1)
      \end{matrix} \right.\right)\nonumber\\
      &\quad \mathrm{H}^{1,1}_{1,1}\left(\frac{1}{\beta_k}x \left|\begin{matrix}
          (q- n,1)\\
          (0,1)
      \end{matrix}\right. \right)\mathrm{H}^{0,1}_{1,0}\left(\frac{1}{\beta_0}x \left|\begin{matrix}
          (1,1)\\
          -
      \end{matrix}\right. \right)dx.
\end{align}
It is exciting to know that the integral in \eqref{eq:20} can be rewritten in terms of the bi-variate Fox-H function \cite[eq. 2.3]{Mittal1972}.

Meanwhile, the second probability in \eqref{eq:18} is derived as
\begin{align}
    I_2 = 1- F_{\tau_j}\left(( \bar{\gamma}_{b} + 1)\frac{\pi_j}{\Psi}\right).
\end{align}

Pulling all the results of $I_1$ and $I_2$ into \eqref{eq:18} and \eqref{eq:16}, we obtain the outcome in \eqref{eq:17} after some manipulations.   
\end{proof}

\section{Numerical Results and Discussions}\label{sec4}
This section provides a few numerical results to verify the information-theoretic framework developed in Sec.~\ref{sec3}, using fixed parameters\footnote{It should be noted that the values of parameters were chosen randomly to show the reality of the results obtained. Thus, based on specific situations, these values can be differently standardized to reflect the required outcomes.}: $\lambda_0 =\Omega_2 =0.25,\lambda_1 =\Omega_1 =0.5, \lambda_2=0.75$,  $\mu = 0.8$, $\alpha_1 = 0.6 \alpha_c$, and $\alpha_2 = 0.4 \alpha_c$. 

Provided that $R_c=0.5$ bps/Hz, $R_1 = R_b = 1$ bps/Hz, $R_2 = 1.5$ bps/Hz, and $\alpha_c = 0.5$. For different transmitting SNR $\Psi$, Fig.~\ref{fig2}(a) confirms an excellent agreement between the theoretical result using \eqref{eq:17} and the actual outcome. As expected from Remark~\ref{remark2}, the SCS provides the worst SOP performance while CCS and MCS are the best with the relative SOP performance to each other. Compared to the benchmark schemes, where the reflection coefficient $\delta$ is fixed with $\delta= 0.3, 0.8$, and the respective SOPs are evaluated based on \eqref{eq:12} and \eqref{eq:13} with $\mathrm{SOP} = 1 -\prod_k\Pr[ \gamma_{k,c} > \bar{\gamma}_{c}, \gamma_{k,k} > \bar{\gamma}_{k}, \gamma_{k,b} > \bar{\gamma}_{b}]$, the findings observed in Fig.~\ref{fig2} show significant SOP improvement as $\Psi$ is increased. All these observations shed light on the efficacy of the developed information-theoretic framework and the proposed GC strategies. 

Next, for different the number of antennas $L$ and transmit SNR, the findings observed in Fig.~\ref{fig3} show that as $L$ increases from $3$ to $6$, the SOP improvement is strongly pronounced with over $3$-dB improvements of $\Psi$ at $\mathrm{SOP }=10^{-3}$. This means the system can save over half of the transmit power $p$. 

Finally, Fig.~\ref{fig4} plots the SOP as a function of rate requirements subject to $L = 4$. As observed, different rate targets produce different SOP trends. Even though Cases 1 and 2 have higher settings of $R_c$ and $R_k$ than Case 3, they produce a better SOP improvement due to a smaller requirement of $R_b$. Compared to Case 3, Case 1 achieves $5$ dB improvement of $\Psi$ while Case 2 offers $2.5$ dB. 
\begin{figure}[t!]
\centering
    \includegraphics[width = 0.869\linewidth]{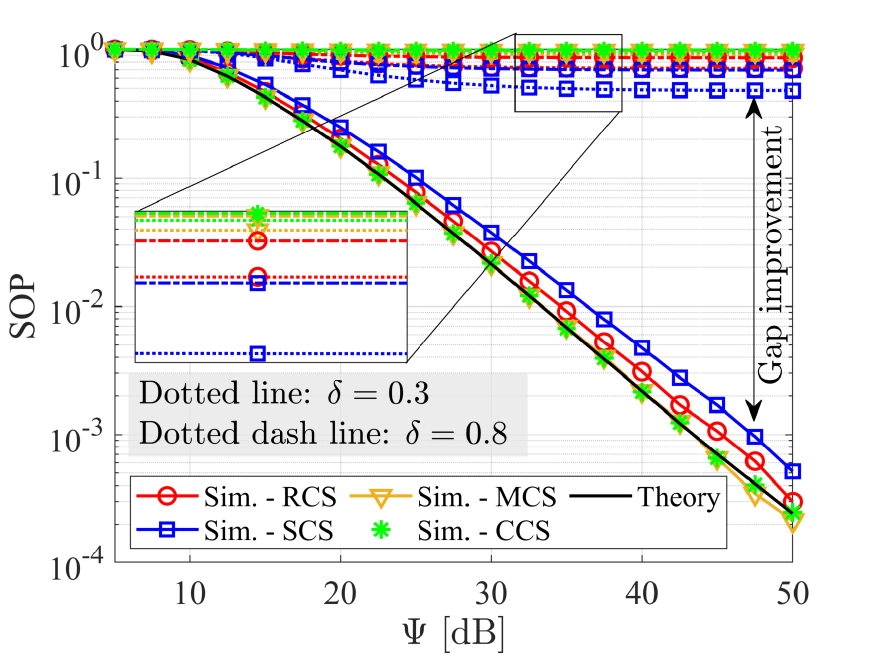}
    \caption{Impact of $\Psi$.}
    \label{fig2}
    \includegraphics[width = 0.869\linewidth]{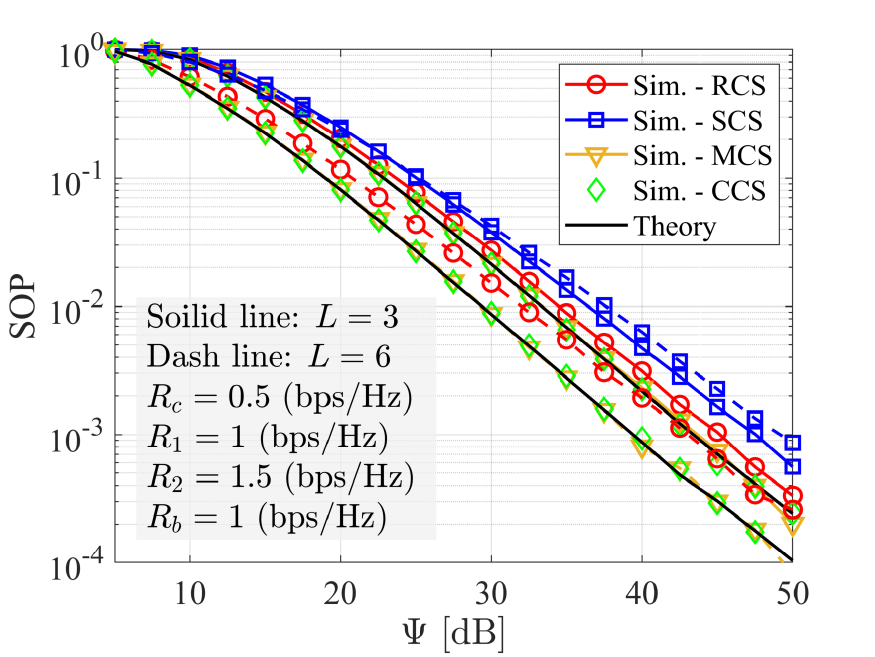}
    \caption{Impact of $K$.}
    \label{fig3}
    \includegraphics[width = 0.869\linewidth]{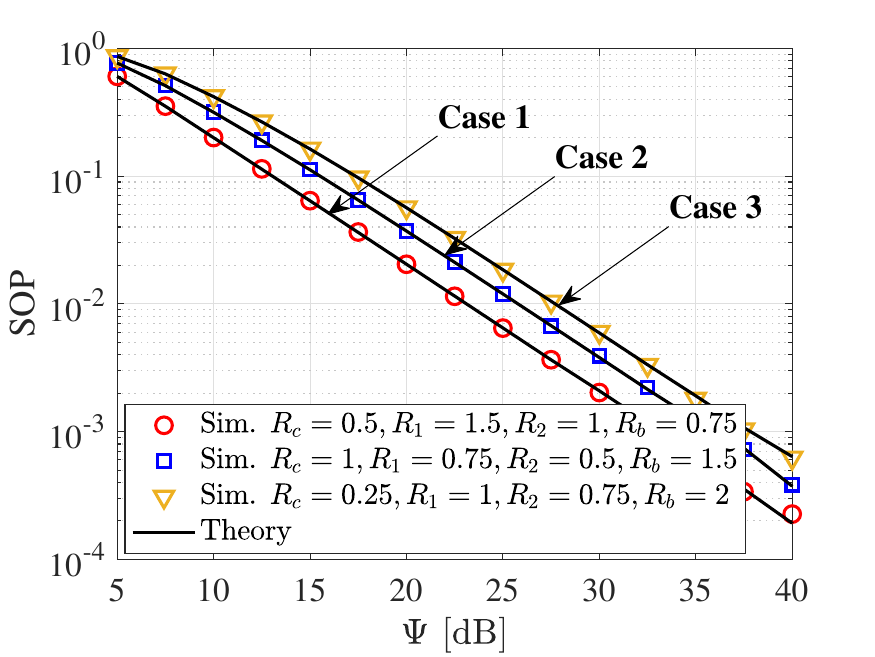}
    \caption{Impact of $R_b,R_k$, and $R_b$.}
    \label{fig4}
\end{figure}

\section{Conclusions and Outlooks}\label{sec5}
A novel paradigm shift of symbiotic backscatter RSMA systems was presented in this contribution, with a series of beamforming weight designs and four potential strategies for system performance enhancement. In particular, a generalized information-theoretic performance framework with a new metric, SOP, had been provided. Numerical results not only validated the developed framework by Monte-Carlo simulations but also showed the promise of the proposed paradigm to the next-generation wireless communication deployment. 

As this is the first proposal, there are still many rooms to be explored in next future. For example, how to allocate the resources to achieve the minimum SOP? How to reduce the impact of imperfect SIC on the performance? How much ergodic capacity can this paradigm deliver? Is there any way to reduce the computation complexity of beamforming designs to deal with real-time response?  These raised questions will open up exciting new directions for future investigations.

%

\end{document}